\def\draft{0}

\documentclass[11pt]{article}
\usepackage{amsfonts, amsmath, amssymb, url, graphics, fullpage}
\usepackage[dvips]{graphicx}

\newtheorem{theorem}{Theorem}

\newtheorem{lemma}{Lemma}
\newtheorem{proposition}{Proposition}
\newtheorem{corollary}{Corollary}

\newtheorem{remk}{Remark}

\def\FullBox{\hbox{\vrule width 8pt height 8pt depth 0pt}}

\def\qed{\ifmmode\qquad\FullBox\else{\unskip\nobreak\hfil
\penalty50\hskip1em\null\nobreak\hfil\FullBox
\parfillskip=0pt\finalhyphendemerits=0\endgraf}\fi}

\newenvironment{proof}{\begin{trivlist} \item {\bf Proof:~~}}
  {\qed\end{trivlist}}

\def\qedsketch{\ifmmode\Box\else{\unskip\nobreak\hfil
\penalty50\hskip1em\null\nobreak\hfil$\Box$
\parfillskip=0pt\finalhyphendemerits=0\endgraf}\fi}

\newcommand{\ie} {{\it i.e.\ }}
\newcommand{\eg} {{\it e.g.\ }}

\newcommand{\N}{{\mathbb{N}}}

\newcommand{\zo}{\{0,1\}}

\newcommand{\pr}[1]{\Pr\left[#1\right]}

\newcommand{\eps}{\varepsilon}

\ifnum\draft=1
\newcommand{\authnote}[2]{{ \bf [#1's Note: #2]}}
\else
\newcommand{\authnote}[2]{}
\fi

\newcommand{\COMMENT}[1]{}

\def\01{\{0,1\}}
\def\01{\{0,1\}}

\newcommand{\cadre}[1]
{
\begin{tabular}{|p{15.4cm}|}
\hline
#1 \\
\hline
\end{tabular}
}

\newcommand{\Abort}{{\mathrm{Abort}}}

\title{Optimal quantum strong coin flipping}
\author{Andr\'e Chailloux$^*$ \\
LRI\\
Universit\'e Paris-Sud\\
andre.chailloux@lri.fr\\
\and
Iordanis Kerenidis\thanks{Supported in part by ANR CRAQ and AlgoQP grants of the French Ministry and in part by the European Commission under the Integrated Project Qubit Applications (QAP) funded by the IST directorate as Contract Number 015848.}\\
CNRS - LRI\\
Universit\'e Paris-Sud\\
jkeren@lri.fr}

\begin{document}

\maketitle

\begin{abstract}
Coin flipping is a fundamental cryptographic primitive that enables two distrustful and far apart parties to create a uniformly random bit~\cite{Blu81}. Quantum information allows for protocols in the information theoretic setting where no dishonest party can perfectly cheat. The previously best-known quantum protocol by Ambainis achieved a cheating probability of at most $3/4$ \cite{Amb01}. On the other hand, Kitaev showed that no quantum protocol can have cheating probability less than $1/\sqrt{2}$ \cite{Kit03}. Closing this gap has been one of the important open questions in quantum cryptography.

In this paper, we resolve this question by presenting a quantum strong coin flipping protocol with cheating probability arbitrarily close to $1/\sqrt{2}$. 
More precisely, we show how to use any weak coin flipping protocol with cheating probability $1/2+\eps$ in order to achieve a strong coin flipping protocol with cheating probability $1/\sqrt{2}+O(\eps)$. The optimal quantum strong coin flipping protocol follows from our construction and the optimal quantum weak coin flipping protocol described by Mochon \cite{Moc07}.

\end{abstract}

\section{Introduction}

Coin flipping is a cryptographic primitive that enables two distrustful and far apart parties, Alice and Bob, to create a random bit that remains unbiased even if one of the players tries to force a specific outcome. It was first proposed by Blum ~\cite{Blu81} and has since found numerous applications in two-party secure computation.
In the classical world, coin flipping is possible under computational assumptions like the hardness of factoring or the discrete log problem. However, 
in the information theoretic setting, it is not hard to see that in any classical protocol, one of the players can always bias the coin to his or her desired outcome with probability 1.  

Quantum information has given us the opportunity to revisit information theoretic security in cryptography. The first breakthrough result was a protocol of Bennett and Brassard~\cite{BB84} that showed how to securely distribute a secret key between two players in the presence of an omnipotent eavesdropper. Thenceforth, a long series of work has focused on which other cryptographic primitives are possible with the help of quantum information. Unfortunately, the subsequent results were not positive. Mayers and Lo, Chau proved the impossibility of secure quantum bit commitment and oblivious transfer and consequently of any type of two-party secure computation~\cite{May97,LC97,DKSW07}. However, several weaker variants of these primitives have been shown to be possible~\cite{HK04,BCH+08}.

The case of coin flipping is one of the most intriguing ones. Even though the results of Mayers and of Lo and Chau exclude the possibility of perfect quantum coin flipping, i.e. where the resulting coin is perfectly unbiased, it still remained open whether one can construct a quantum protocol where no player could bias the coin with probability 1. A few years later, Aharonov et al.~\cite{ATVY00} provided such a protocol where no dishonest player could bias the coin with probability higher than 0.9143. Then, Ambainis~\cite{Amb01} described an improved protocol whose cheating probability was at most $3/4$. Subsequently, a number of different protocols have been proposed~\cite{SR01,NS03,KN04} that achieved the same bound of $3/4$.

On the other hand, Kitaev~\cite{Kit03}, using a formulation of quantum coin flipping protocols as semi-definite programs proved a lower bound of $1/2$ on the product of the two cheating probabilities for Alice and Bob (for a proof see \eg ~\cite{ABD+04}). In other words, no quantum coin flipping protocol can achieve a cheating probability less than $1/\sqrt{2}$ for both Alice and Bob. 

The question of whether $3/4$ or $1/\sqrt{2}$ is ultimately the right bound for quantum coin flipping has been open since then. In fact, there had been ``evidence'' suggesting both cases. First, Kitaev's semi-definite program formulation of coin flipping seems to be a natural one and using this semi-definite program one cannot hope to prove a better lower bound. On the other hand, most of the suggested coin flipping protocols were using some form of imperfect bit commitment scheme. More precisely, Alice would quantumly commit to a bit $a$, Bob would announce a bit $b$ and then Alice would reveal her bit $a$. The outcome of the coin flip would be $a \oplus b$. However, Ambainis had proved a lower bound of $3/4$ for any protocol of this type and even though more complicated protocols based on similar ideas had been proposed, they all seemed to get stuck at the same $3/4$ bound.

During the study of quantum coin flipping, a weaker variant was introduced that is referred to as {\em weak coin flipping} in opposition to the original \emph{strong coin flipping}. In this setting, Alice and Bob have a priori a desired coin outcome, in other words the two values of the coin can be thought of as `Alice wins' and `Bob wins'. We are again interested in bounding the probability that a dishonest player can win this game.  

Weak coin flipping protocols with cheating probabilities less than $3/4$ were constructed in ~\cite{SR02,Amb02,KN04}. The best achieved bound was in fact $1/\sqrt{2}$, a strange coincidence, since Kitaev's lower bound of $1/\sqrt{2}$ does not apply in the case of weak coin flipping. The only lower bound that carries over from the case of strong coin flipping is a bound by Ambainis that shows that in order to achieve a cheating probability of $1/2+\eps$ the protocol must have at least $O(\log \log \frac{1}{\eps})$ rounds~\cite{Amb02}. We refer to $\eps$ as the bias of the protocol.

Finally, a breakthrough result by Mochon resolved the question of the optimal quantum weak coin flipping. First, he described a protocol with cheating probability $2/3$~\cite{Moc04,Moc05} and then a protocol that achieves a cheating probability of $1/2+\eps$ for any $\eps>0$ ~\cite{Moc07}. Kitaev's formalism and Mochon's optimal weak coin flipping protocol delve heavily into the theory of convex cones and operator monotone functions.

In this work, we resolve the question of the optimal quantum strong coin flipping protocol. We present a general method on how to use any weak coin-flipping protocol with cheating probability $1/2+\varepsilon$ in order to construct a strong coin-flipping protocol with cheating probability $1/\sqrt{2}+O(\eps)$. Our protocol uses roughly the same number of rounds as the weak coin flipping protocol. Combining our construction with Mochon's quantum weak coin flipping protocol that achieves arbitrarily small bias, we conclude that it is possible to construct a quantum strong coin flipping protocol with cheating probability arbitrarily close to $\frac{1}{\sqrt{2}}$. 

Let us make a few remarks about our protocol. First, it is a {\em classical} protocol that uses a weak coin flipping as a subroutine. In other words, in coin flipping, the power of quantum really comes from the ability to perform weak coin flipping. If there existed a classical weak coin flipping protocol with arbitrarily small bias, then this would have implied a classical strong coin flipping protocol with cheating probability arbitrarily close to $1/\sqrt{2}$ as well.
Moreover, our protocol has the advantages of being very easy to describe and having a straightforward analysis, assuming, of course, the existence of a weak coin flipping protocol. 

Using weak coin flipping in order to perform strong coin flipping is not a new idea. 
There is a trivial protocol that uses a perfect weak coin flipping and achieves strong coin flipping with cheating probability $3/4$: Alice and Bob run the weak coin flipping protocol and whoever wins, flips a random coin $c \in_R \zo$.

Our protocol can be thought of as a refinement of the abovementioned one. There are two simple ideas that we use. First, we will have Alice flip and announce the outcome of her random coin {\em before} Alice and Bob perform the weak coin flipping protocol. 
Second, we will use an ``unbalanced'' weak coin flipping, where in the honest case, Alice wins with probability $z$ and Bob with probability $1-z$. 

We can now describe informally our protocol  \\ \\
\cadre{
\begin{center} Strong Coin Flipping Protocol \end{center}
\begin{itemize}
\item Alice flips a random coin $a$ and sends $a$ to Bob. 
\item Alice and Bob run an unbalanced weak coin flipping protocol where honest Alice wins with probability $z$ and honest Bob with probability $1-z$. 
\item If Alice wins, the output is $a$.
\item If Bob wins, he outputs $a$ with probability $p$ and $\overline{a}$ with probability $1-p$.
\end{itemize}
}
$ \ $ \\\\

In Section~\ref{analysis}, we first show how to construct ``unbalanced'' weak coin flipping protocols for any $z$ and bias $O(\eps)$, assuming the existence of a ``balanced'' weak coin flipping protocol with bias $\eps$. Then, we optimize the parameters $z$ and $p$ in order to make the cheating probability of our protocol at most $1/\sqrt{2}+O(\eps)$.

\section{Definitions} \label{Preliminaries}

We provide the formal definitions of all the different variants of coin flipping protocols that we are going to use.

A coin flipping protocol between two parties Alice and Bob is a protocol where Alice and Bob interact and at the end, Alice outputs a value $c_A \in \{0,1,\Abort\}$ and Bob outputs a value $c_B \in \{0,1,\Abort\}$. If $c_A = c_B$, we say that the protocol outputs $c = c_A$. If $c_A \neq c_B$ then the protocol outputs $c = \Abort$.

In a coin flipping protocol, we call a round of communication one message from Alice to Bob and one message from Bob to Alice. We suppose that Alice always sends the first message and Bob always sends the last message. The protocol is quantum if we allow the parties to send quantum messages and perform quantum operations. A player is honest if he or she follows the protocol. A cheating player can deviate arbitrarily from the protocol but still outputs a value at the end of it.
There are two important variants of coin flipping that have been studied.

\paragraph{Weak coin flipping\\}
A (balanced) weak coin flipping protocol with bias $\varepsilon$ ($WCF(1/2,\varepsilon)$) has the following properties
\begin{itemize}
\item If $c = 0$, we say that Alice wins. If $c = 1$, we say that Bob wins.
\item If Alice and Bob are honest then $\pr{\mbox{ Alice wins }} = \pr{ \mbox{ Bob wins }} = 1/2$
\item If Alice cheats and Bob is honest then $P^*_A = \pr{ \mbox{ Alice wins }} \le 1/2 + \varepsilon$
\item If Bob cheats and Alice is honest then $P^*_B = \pr{ \mbox{ Bob wins }}\le 1/2 + \varepsilon$
\end{itemize}
The probabilities $P^*_A$ and $P^*_B$ are called the cheating probabilities of Alice and Bob respectively. The cheating probability of the protocol is defined as $\max\{P^*_A,P^*_B\}$. We say that the coin flipping is \emph{perfect} if $\eps = 0$.

 We can also define weak coin flipping for the case where the winning probabilities of the two players in the honest case are not equal.

\paragraph{Unbalanced weak coin flipping\\} 
A weak coin flipping protocol with parameter $z$ and bias $\varepsilon$ ($WCF(z,\varepsilon))$ has the following properties.

\begin{itemize}
\item If $c = 0$, we say that Alice wins. If $c = 1$, we say that Bob wins.
\item If Alice and Bob are honest then $\pr{\mbox{ Alice wins }} = z$ and $\pr{\mbox{ Bob wins }} = 1 - z$
\item If Alice cheats and Bob is honest then $P^*_A = \pr{\mbox{ Alice wins }} \le z + \varepsilon$
\item If Bob cheats and Alice is honest then  $P^*_B = \pr{\mbox{ Bob wins }} \le (1 - z) + \varepsilon$
\end{itemize}

\paragraph{Strong coin flipping\\} A strong coin flipping protocol with bias $\varepsilon$ ($SCF(\varepsilon)$) has the following properties
\begin{itemize}
\item If Alice and Bob are honest then $\pr{c = 0} = \pr{c=1} = 1/2$
\item If Alice cheats and Bob is honest then $P^*_A = \max\{\pr{c = 0},\pr{c = 1}\} \le 1/2 + \varepsilon$.
\item If Bob cheats and Alice is honest then $P^*_B = \max\{\pr{c = 0},\pr{c = 1}\} \le 1/2 + \varepsilon$
\end{itemize} 
Similarly, $P^*_A$ and $P^*_B$ are the cheating probabilities of Alice and Bob. The cheating probability of the protocol is defined as $\max\{P^*_A,P^*_B\}$. 

We will also use the following result by Mochon

\begin{proposition}{\em\cite{Moc07}}\label{Mochon}
For every $\varepsilon > 0$, there exists a quantum $WCF(1/2,\varepsilon) $ protocol $P$. 
\end{proposition}


\section{An optimal strong coin flipping protocol} \label{Protocol}

In this section, we describe how to construct an optimal strong coin flipping protocol from any weak coin flipping protocol. Let us try to give some intuition for our protocol before we actually describe and analyze it. For this high level discussion, we assume the existence of a perfect weak coin flipping protocol. As we said, there exists a trivial protocol that uses a weak coin flipping in order to achieve strong coin flipping: 

 \begin{center}\textbf{$\mathbf{SCF(3/4)}$ protocol using a perfect weak coin flipping protocol $P$} \end{center}
\begin{itemize}
\item Alice and Bob run the protocol $P$
\item The winner chooses a random $c \in_R \zo$, and sends $c$ to the other player, $c$ being the outcome of the protocol.  
\end{itemize}  
 
Let us analyze this protocol more closely. Let Alice be dishonest and her desired value for the coin be 0. Her strategy will be to try and win the WCF protocol, which happens with probability $1/2$ and then output 0. However, even if she loses the weak coin flipping, there is still a probability $1/2$ that the honest Bob will output 0. Hence, Alice's (and by symmetry Bob's) cheating probability is $3/4$. 

In order to reduce this bias, we would like to eliminate the situation where the honest player, after winning the WCF, still helps the dishonest player cheat with probability 1/2. One can try to resolve this problem by having Alice flip and announce her random coin $c$ before running the WCF protocol. In this case: first, Alice announces a bit $a$. Then, Alice and Bob perform a WCF. If Alice wins the outcome is $a$; if Bob wins then the outcome is $\overline{a}$.

In this case, Bob never outputs $a$. However, there is a simple cheating strategy for Alice. If she wants 0, she sets $a=1$, loses the WCF (which she can do with probability 1) and therefore Bob always outputs 0. Hence, Bob's choice when he wins the WCF must be probabilistic. Let us now consider the following protocol: 
\begin{center}\textbf{Improved $SCF$ protocol using a perfect weak coin flipping protocol $P$} \end{center}
\begin{itemize}
\item Alice picks a random bit $a \in_R \zo$ and sends $a$ to Bob.
\item Alice and Bob run $P$.
\item If Alice wins then the outcome is $a$. 
\item If Bob wins then he outputs $a$ with probability $p$ and $\overline{a}$ with probability $1-p$. 
\end{itemize}  

First note that Bob's cheating probability is $3/4$, independent of $p$. Namely, if Alice picks the value Bob wants, then he just loses the WCF; if Alice picks the opposite value then he tries to win the WCF in order to pick his desired value. 
On the other hand, let us calculate Alice's cheating probability. Alice can pick $a$ to be equal to her desired outcome, in which case the final outcome is $a$ with probability $\frac{1}{2}\cdot 1 + \frac{1}{2}\cdot p$. She may also pick $a$ to be the opposite of her desired outcome, in which case she always loses the WCF and hence, the final outcome is $\overline{a}$ with probability $1-p$. Alice's cheating probability is the maximum of the two cases. By choosing $p=1/3$ the probabilities in the two cases are equal and we conclude that Alice's cheating probability is $2/3$.

Hence, using any balanced WCF protocol, we can have a strong coin flipping protocol that achieves cheating probability $3/4$ for Bob and $2/3$ for Alice. This in some sense already beats the best previously known protocol that can only achieve cheating probability $3/4$ for both Alice and Bob. 

Our next step is to make these two cheating probabilities equal. For this, we will use an unbalanced WCF with parameter $z$ and optimize $z$ and $p$ in order to get the optimal bound of $1/\sqrt{2}$.

We can now describe our final protocol that uses a $WCF(z,\eps)$ protocol $Q$ as a subroutine. 
\\ \\
\cadre{\begin{center} \textbf{Strong Coin Flipping protocol $S$} \end{center}
\begin{enumerate}
\item Alice chooses $a \in_R \zo $ and sends $a$ to Bob.
\item Alice and Bob perform the $WCF(z,\varepsilon)$ protocol $Q$
\begin{itemize}
\item If Alice wins $Q$ then honest players output $c_A = c_B = a$
\item If Bob wins $Q$ then he flips a coin $b$ such that $b = a$ with probability $p$ and $b = \overline{a}$ with probability $(1-p)$. He sends $b$ to Alice. In this case, honest players output $c_A = c_B = b$.
\item If $Q$ outputs Abort then Abort
\end{itemize}
\end{enumerate}
}\\\\


\section{Analysis of the strong coin flipping protocol $S$} \label{analysis}

We first describe the construction of the unbalanced $WCF(z,\eps)$ protocol and then show how to optimize the parameters $z$ and $p$  in order to achieve the optimal bias.


\subsection{An unbalanced weak coin flipping protocol}\label{Unbalanced}

Our goal is to prove the following proposition

\begin{proposition}~\label{Approx}
Let  $P$ be a $WCF(1/2,\varepsilon)$ protocol with $N$ rounds. Then, 
$\forall z \in [0,1]$ and $\ \forall k \in \N$, there exists a  $WCF(x,\varepsilon_0)$ protocol $Q$ such that:
\begin{itemize}
\item $Q$ uses $k\cdot N$ rounds.
\item $|x - z| \le 2^{-k}$.
\item $\varepsilon_0 \le 2\varepsilon$.
\end{itemize}
\end{proposition}

The protocol $Q$ is a sequential composition of the $WCF(1/2,\eps)$ protocol $P$. In high level,
we use $P$ in order to combine two weak coin flipping protocols with parameters $z_1$ and $z_2$ into a new protocol with parameter $\frac{z_1 + z_2}{2}$. Then, by recursion, for any given $z$ we can create a protocol $Q$ with parameter $x$ that rapidly converges to $z$. We also prove that the bias of $Q$ is at most $2\eps$. 

Assume we have a $WCF(z_1,\eps_0)$ protocol $P_1$ and a $WCF(z_2,\eps_0)$ protocol $P_2$ each with at most $M$ rounds of communication and $z_2 \geq z_1$. We combine them in the following way. 

\begin{center} $\mathbf{Comb(P_1,P_2)}$ \end{center}
\begin{itemize} 
\item Alice and Bob run $P$. 
\item If Alice wins, run $P_2$. If Bob wins, run $P_1$. If $P$ Aborts then Abort.
\end{itemize} $ \ $ \\
Note that this protocol uses at most $N+M$ rounds. We have
\begin{lemma}\label{comb}
$Comb(P_1,P_2)$ is a $WCF(\frac{z_1 + z_2}{2},\varepsilon_0 + \varepsilon(z_2 - z_1))$ protocol.
\end{lemma}

\begin{proof} \\
{\bf Alice and Bob are honest}
If Alice and Bob are honest then the protocol never aborts. We have Pr[ Alice wins ] $= \frac{z_1 + z_2}{2}$ and  Pr[ Bob wins ] $= 1 - \frac{z_1 + z_2}{2}$. 

\paragraph{Alice cheats and Bob is honest} Let $x = $ Pr[Alice wins  P] ; $y = $ Pr[Bob wins  P]; $u = $ Pr[Alice wins $ P_2 \ |$ Alice wins P]; $v = $ Pr[Alice wins  $P_1 \ |$ Bob wins P].  We know the following inequalities concerning these probabilities: 
\[
x + y \le 1 \ \ \ \ x \le 1/2 + \varepsilon \ \ \ \ u \le z_2 + \varepsilon_0 \ \ \ \ v \le z_1 + \varepsilon_0
\]

Note that the last two inequalities hold, since the biases for the protocols $P_1$ and $P_2$ do not increase depending on the outcome of $P$. We have
\begin{eqnarray*}
\lefteqn{\pr{ \mbox{ Alice wins } Comb(P_1,P_2)}}\\ & = & x\cdot u + y \cdot v
\;\; \le  \;\; x(z_2 + \varepsilon_0) + (1-x)(z_1 + \varepsilon_0) 
\;\; = \;\; (z_1 + \varepsilon_0) + x(z_2 - z_1)  \\
 & \le & (z_1 + \varepsilon_0) + (1/2 + \varepsilon)(z_2 - z_1) \qquad \textrm{since } z_2 \ge z_1 \\
& \le & \frac{z_1 + z_2}{2} + \varepsilon_0 + \varepsilon(z_2 - z_1) 
\end{eqnarray*}

\paragraph{Bob cheats and Alice is honest} Using a similar calculation as in the previous case, we have $\textrm{Pr[Bob wins } Comb(P_1,P_2)] \le  \frac{(1 - z_2) + (1 - z_1)}{2} + \varepsilon_0 + \varepsilon(z_2 - z_1) = 1 - \frac{z_1 + z_2}{2} + \varepsilon_0 + \varepsilon(z_2 - z_1) $.
\end{proof}
We now show the following inductive Lemma 

\begin{lemma} \label{Induction}
Suppose we have a $WCF(1/2,\varepsilon)$ protocol $P$ that uses $N$ rounds of communication. Then $\forall z \in [0,1] $ and $ \forall k \in \N,$ 
we can construct a $WCF(x_1,\varepsilon_0)$ protocol $P_1$ and a $WCF(x_2,\varepsilon_0)$ protocol $P_2$ such that
\begin{itemize}
\item $P_1,P_2$ each use at most $k\cdot N$ rounds.
\item $ x_1 \le z \le x_2 \;$ and $ \;x_2 - x_1 = 2^{-k}$.
\item $\varepsilon_0 \le (2 - 2(x_2 - x_1))\varepsilon$.
\end{itemize}
\end{lemma}
\begin{proof}
Fix $z \in [0,1]$. We show this result by induction on $k$. For $k=0$, we clearly have a $WCF(0,0)$ protocol (a protocol where Bob always wins) and a $WCF(1,0)$ (a protocol where Alice always wins) that use no rounds of communication. We suppose the Lemma is true for $k$ and we show it for $k+1$.

Let $x_1,x_2,P_1,P_2,\varepsilon_0$ that satisfy the above properties for $k$. Let $P'$ be the $Comb(P_1,P_2)$ protocol and $u = \frac{x_1 + x_2}{2}$. $P'$ uses at most $(k+1)N$ rounds and from Lemma \ref{comb}, we know that $P'$ is a $WCF(u, \varepsilon'_0 = \varepsilon_0 + (x_2 - x_1)\varepsilon)$ protocol. From the induction step we have that $\varepsilon'_0 \le (2 - 2(x_2 - x_1))\varepsilon + (x_2 - x_1)\varepsilon \le (2 - (x_2 - x_1))\varepsilon$. We now distinguish two cases
\begin{itemize}
\item If $z \le u$, consider the protocols $P_1$ and $P'$. Each one uses at most $(k+1)N$ rounds. Also, $x_1 \leq z \leq u$ and $u - x_1 = \frac{x_2 - x_1}{2} = 2^{-(k+1)}$. Finally, $\varepsilon'_0 \le (2 - (x_2 - x_1))\varepsilon = (2 - 2(u - x_1))\varepsilon$ which concludes the proof.
\item If $z > u$, consider the protocols $P'$ and $P_2$. Each one uses at most $(k+1)N$ rounds. Also, $u \leq z \leq x_2$ and $x_2 - u = \frac{x_2 - x_1}{2} = 2^{-(k+1)}$. Finally, $\varepsilon'_0 \le (2 - (x_2 - x_1))\varepsilon = (2 - 2(x_2 - u))\varepsilon$ which concludes the proof.
\end{itemize}
\end{proof}

In Lemma~\ref{Induction}, we have $|x_1 - z| \leq (x_2-x_1) \le 2^{-k}$ and $\varepsilon_0 \le 2\eps$. Hence this Lemma directly implies Proposition~\ref{Approx} by considering $Q = P_1$.


\subsection{Strong coin flipping from unbalanced weak coin flipping}\label{Main}

We calculate the cheating probability of our protocol $S$ that uses a $WCF(z,\eps)$ protocol $Q$. 
\begin{proposition}\label{propSCF}
The protocol $S$ is a strong coin flipping protocol with $N + 2$ rounds of communication and cheating probabilities $P^*_A \le \frac{1}{2 - z - \varepsilon}$ and $P^*_B \le \frac{2 - z + \varepsilon}{2}$.
\end{proposition}

\begin{proof}\\
{\bf Alice and Bob are honest}
If both players are honest then they never abort. Moreover, since the protocol is symmetric in $0$ and $1$, we have $\pr{c=0} = \pr{c=1} = 1/2 $.

\noindent
{\bf Alice cheats and Bob is honest}
We prove that $\pr{c = 0} \le \frac{1}{2 - z - \varepsilon}$. By symmetry, the same holds for $\pr{c=1}$.
Since Alice cheats, she can choose arbitrarily between $a = 0$ and $a = 1$ instead of picking $a$ uniformly at random. Hence, $Pr[c = 0] \le \max\{\pr{ c = 0 | a = 0},\pr{ c = 0 | a = 1}\}$.
\begin{itemize}
\item We first calculate $\pr{c = 0 | a = 0}$. \\ Let $x = \pr{\mbox{Alice wins } Q | a~=~0}$ and $y = \pr{\mbox{Bob wins } Q | a~=~0}$. We have $\pr{c = 0 | a = 0} = x \cdot 1 + y\cdot p$. Note that $x + y \le 1$ and also $x \le z + \varepsilon$, since the maximum bias with which Alice can win $Q$ is independent of the value of $a$. We have
\begin{eqnarray*}
\pr{c = 0 |  a = 0 } & = & x \cdot 1 + y\cdot p 
\;\; \le \;\; x + (1-x)p 
\;\; = \;\; p + x(1 - p) \\
& \le & p + (z + \varepsilon)(1-p)
\end{eqnarray*}
\item We now calculate $\pr{c = 0 | a = 1}$. \\ Let $x = \pr{ \mbox{Alice wins } Q | a = 1}$ and $y = \pr{ \mbox{Bob wins } Q | a = 1}$. We have 
\[ \pr{c = 0 | a~=~1} = x\cdot 0 + y(1 - p) \le y(1 - p) \le 1 - p \]
which is achievable since Alice could always let Bob win $Q$.
\end{itemize}
Since $\pr{c = 0} \le \max\{\pr{ c = 0 | a = 0},\pr{ c = 0 | a = 1}\}$, we choose $p$ such that the upper bounds for $\pr{ c = 0 | a = 0}$ and $\pr{ c = 0 | a = 1}$ are equal.
\begin{eqnarray*}
p + (z + \varepsilon)(1-p) & = & 1 - p \\
p & = & \frac{1 - z - \varepsilon}{2 - z - \varepsilon} 
\end{eqnarray*}
With this value of $p$, we have 
\begin{eqnarray*}
Pr[c = 0] & \le & \max\{\pr{ c = 0 | a = 0},\pr{ c = 0 | a = 1}\}
\;\; = \;\; 1 - p 
\;\; \le \;\; \frac{1 }{2 - z - \varepsilon} 
\end{eqnarray*}
Since the protocol is symmetric in $0$ and $1$, we also have $\pr{c = 1} \le \frac{1 }{2 - z - \varepsilon} $ and hence $P^*_A \le \frac{1 }{2 - z - \varepsilon} $. \\

\noindent
{\bf Bob cheats and Alice is honest}
We prove that $\pr{c = 0} \le \frac{2 - z + \varepsilon}{2}$. By symmetry, the same holds for $\pr{c = 1}$.
Alice is honest and picks $a$ uniformly at random. We first have $\pr{c = 0 | a = 0}  \le 1$.
We now upper bound $\pr{c = 0 | a = 1}$. Let $x = \pr{ \mbox{Bob wins } Q | a = 1}$ and $y = \pr{ \mbox{Alice wins } Q | a = 1}$. We have 
\[ \pr{c = 0 | a = 1} \le x\cdot 1 + y\cdot 0 \le x \le 1 - z + \varepsilon
\]
Since Alice is honest, we have $\pr{a = 0} = \pr{a = 1} = 1/2$ and hence:
\begin{eqnarray*}
\pr{c = 0 } & = & \pr{c = 0 | a = 0 }\cdot \pr{ a = 0 } + \pr{ c = 0 | a = 1 } \cdot \pr{a = 1} \\
& = & \frac{1}{2}\left( \pr{c = 0 | a = 0 } + \pr{c = 0 | a = 1 } \right) \\
& \le & \frac{1}{2} + \frac{1 - z + \varepsilon}{2} \\
& = & \frac{2 - z + \varepsilon}{2}
\end{eqnarray*}
Since the protocol is symmetric in $0$ and $1$, we also have $\pr{c = 1} \le \frac{2 - z + \varepsilon}{2}$ and hence $P^*_B \le \frac{2 - z + \varepsilon}{2}$. \\

\end{proof}

\subsection{Putting it all together} \label{Conclusion}

To conclude, we have to optimize $z$. In the case where there exists an ideal weak coin flipping protocol $WCF(1/2,0)$, it is easy to see that in order to equalize the cheating probabilities $P^*_A$ and $P^*_B$, we need to take $z=2-\sqrt{2}$. If also our Proposition \ref{Approx} was ideal, \ie if from $P$ we could create perfectly a $WCF(2-\sqrt{2},0)$ protocol $Q$, then $S$ would have cheating probability exactly $\frac{1}{\sqrt{2}}$.

In general, we need to take care of the small bias $\eps$ of the initial $WCF(1/2,\eps)$ protocol $P$ and the error of our Proposition \ref{Approx}. However, we will see that the overall increase in the cheating probability of our protocol $S$ is only $O(\eps)$.

\begin{theorem}\label{SCF}
If there exists a $WCF(1/2,\varepsilon)$ protocol $P$ that uses $N$ rounds of communication then there exists a strong coin flipping protocol $S$ that uses $2\lceil \log(\frac{1}{\eps}) \rceil \cdot N + 2$ rounds with cheating probability at most $\frac{1}{\sqrt{2}} + \sqrt{2}\eps+o(\eps)$.
\end{theorem}

\begin{proof}
Starting from the $WCF(1/2,\eps)$ weak coin flipping protocol $P$ with $N$ rounds, we can use Proposition \ref{Approx} with $k=2\lceil \log(\frac{1}{\eps}) \rceil$ and construct a $WCF(x,\varepsilon')$ protocol $Q$ with the following properties
\begin{itemize}
\item $Q$ uses $2\lceil \log(\frac{1}{\eps}) \rceil \cdot N$ rounds.
\item $| x - (2-\sqrt{2}) | \le \eps^2$.
\item $\varepsilon' \le 2\varepsilon$.
\end{itemize}

Then, we use the protocol $Q$ in the strong coin flipping protocol $S$ we described in Section \ref{Protocol} and by Proposition \ref{propSCF} we have that $S$ has  $2\lceil \log(\frac{1}{\eps}) \rceil\cdot N+2$ rounds and 
\begin{eqnarray*}
P^*_A & = & \frac{1 }{2 - x - \varepsilon'} 
\;\; \le \;\; \frac{1}{\sqrt{2} - 2\eps-\eps^2} \leq \frac{1}{\sqrt{2}} + \sqrt{2}\varepsilon+o(\eps)\\
P^*_B & = & \frac{2 - x + \varepsilon'}{2} 
\; \le \;\; \frac{\sqrt{2} +  2\varepsilon +\eps^2}{2} 
\;\; = \;\; \frac{1}{\sqrt{2}} + \varepsilon+o(\eps)
\end{eqnarray*}
\end{proof}

\noindent
Using Theorem \ref{SCF} and Mochon's weak coin flipping protocol (Proposition \ref{Mochon}) we conclude that
\begin{corollary}
For any $\eps > 0$, there exists a strong coin flipping protocol with cheating probability $\frac{1}{\sqrt{2}} + \eps$.
\end{corollary}

Last, note that our strong coin flipping protocol uses $O(N \cdot\log(\frac{1}{\eps}))$ rounds, where $N$ is the number of rounds of Mochon's weak coin flipping protocol.

\section{Conclusion}
In this paper, we presented the first quantum strong coin flipping protocol with a cheating probability arbitrarily close to the optimal value $\frac{1}{\sqrt{2}}$. Our protocol uses as a subroutine the quantum weak coin flipping protocol designed by Mochon which is arbitrarily close to optimal. Note that except when using this quantum weak coin flipping protocol, our entire protocol is classical.

We would like to note that Mochon's protocol is still not very well understood (protocol's unitary description, number of rounds). It is important to get a better understanding of that protocol and/or find a simpler construction of an optimal quantum weak coin flipping protocol. Moreover, it would be interesting to study what other cryptographic primitives can be derived from weak or strong coin flipping.

\newcommand{\etalchar}[1]{$^{#1}$}

\end{document}